\documentclass{article}
\usepackage{amsmath}
\usepackage{amsfonts}
\usepackage{amssymb}
\usepackage{caption,subcaption}
\usepackage{graphicx}
\usepackage{fullpage}
\usepackage{amsthm}
\usepackage{mathtools,thmtools}
\usepackage{color}
\usepackage{blkarray}
\usepackage{hyperref}

\theoremstyle{definition}
\newtheorem{theorem}{Theorem}

\newtheorem{definition}{Definition}
\newtheorem{example}[definition]{Example}

\newtheorem{lemma}[theorem]{Lemma}

\newtheorem{remark}[definition]{Remark}

\newcommand{\eproj}{\mathrm{Proj}}

\newcommand{\e}{\mathrm{e}}
\renewcommand{\S}{S}
\newcommand{\R}{\mathcal{R}}
\newcommand{\po}{\mathrm{Poisson}}
\renewcommand{\|}{\,||\,}

\newenvironment{bullets}%
        {\begin{list}
                {\noindent\makebox[0mm][r]{$\bullet$}}
                {\leftmargin=5.5ex \usecounter{enumi}
 		 \topsep=1.5mm \itemsep=-.75ex}
        }
        {\end{list}}

\title{A stochastic molecular scheme for an artificial cell to infer its environment from partial observations}

\author{Muppirala Viswa Virinchi \and  Abhishek Behera \and Manoj Gopalkrishnan}

\date{India Institute of Technology Bombay, Mumbai, India \\
{$\{$axlevisu, abhishek.behera.iitm, manoj.gopalkrishnan$\}$@gmail.com}\ 
\\April 1, 2017}

\begin{document}
\maketitle

\begin{abstract}
The notion of entropy is shared between statistics and thermodynamics, and is fundamental to both disciplines. This makes statistical problems particularly suitable for reaction network implementations. In this paper we show how to perform a statistical operation known as Information Projection or E projection with stochastic mass-action kinetics. Our scheme encodes desired conditional distributions as the equilibrium distributions of  reaction systems. To our knowledge this is a first scheme to exploit the inherent stochasticity of reaction networks for information processing. We apply this to the problem of an artificial cell trying to infer its environment from partial observations.
\end{abstract}

\section{Introduction}\label{S:1}
Biological cells function in environments of high complexity. Transmembrane receptors allow a cell to sample the state of its environment, following which biochemical reaction networks integrate this information, and compute decision rules which allow the cell to respond in sophisticated ways. One challenge is that receptors may be imperfectly specific, binding to multiple ligands with various propensities. What algorithmic and statistical ideas are needed to deal with this challenge, and how would these ideas be implemented with reaction networks? These are the questions we begin to address here. The two questions do not decouple because the attractiveness of algorithmic and statistical ideas towards these challenges is tied in with their ease of implementation with reaction networks. We are interested in statistical algorithms that fully exploit the native dynamics and stochasticity of reaction networks. To fix ideas, let us consider an example.
\begin{figure}[h]
\centering
\begin{subfigure}{.5\textwidth}
\centering
\includegraphics[width=0.9\linewidth]{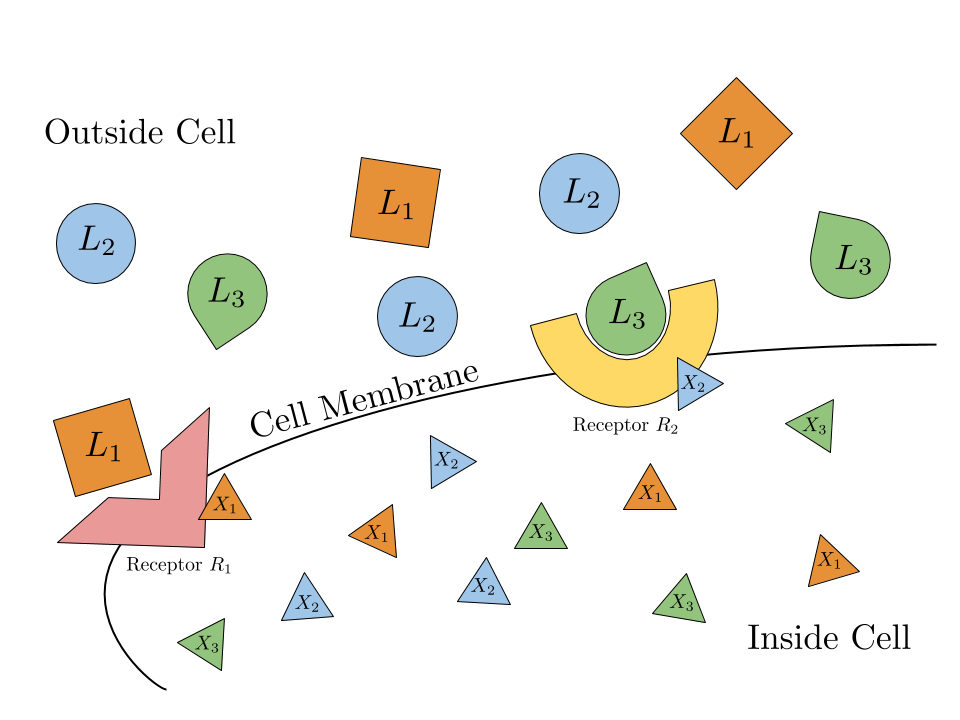}
\end{subfigure}%
\begin{subfigure}{.5\textwidth}
\centering
\includegraphics[width=0.9\linewidth]{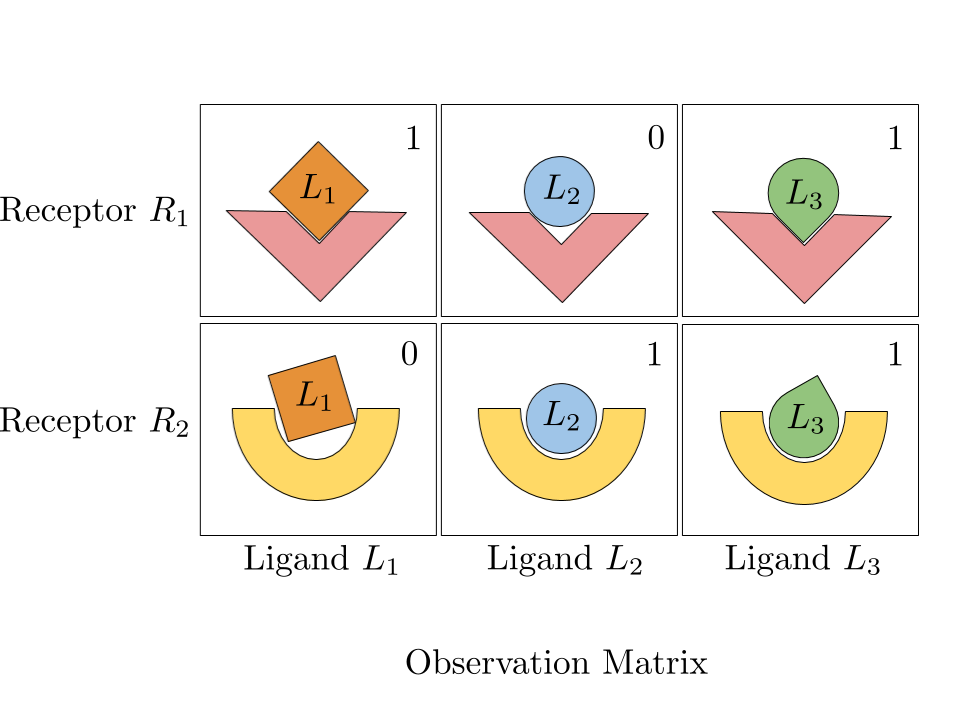}
\end{subfigure}
\caption{\label{fig:1}An artificial cell with two transmembrane receptors $R_1$ and $R_2$ and extracellular ligands $L_1,L_2,L_3$. $R_1$ has equal affinity to both $L_1$ and $L_3$. $R_2$ has equal affinity to both $L_2$ and $L_3$.}
\end{figure}
\begin{example}\label{ex:run}
Consider an artificial cell with two types of transmembrane receptors $R_1$ and $R_2$ in an environment with three ligand species $L_1, L_2$, and $L_3$ (Figure~\ref{fig:1}). Receptor $R_1$ has equal affinity to ligands $L_1$ and $L_3$, and no affinity to $L_2$. Receptor $R_2$ has equal affinity to ligands $L_2$ and $L_3$, and no affinity to $L_1$. This information can be summarized in an \textbf{observation matrix} 
\[
O=\begin{blockarray}{cccc}
&L_1 &L_2 &L_3\\
\begin{block}{c(ccc)}
R_1  &1 &0 &1\\
R_2  &0 &1 &1\\
\end{block}
\end{blockarray}
\]

The question of interest is how to design a cytoplasmic chemical reaction network to estimate the numbers $l_1,l_2,l_3$ of the ligands from receptor binding information. We assume that a prior probability distribution over ligand states $(l_1,l_2,l_3)\in\mathbb{Z}_{\geq 0}^3$ is given. We further assume that this prior probability distribution is a product of Poisson distributions specified by given Poisson rate parameters $q_1,q_2,q_3\in\mathbb{R}_{>0}$ respectively. Lemma~\ref{lem:productpoisson} provides intuition for the product-Poisson assumption. The following questions concern us.
\begin{enumerate}
\item\label{Q1} Given information on the exact numbers $r_1$ and $r_2$ of binding events of receptors $R_1$ and $R_2$, obtain samples over populations $(l_1,l_2,l_3)$ of the ligand species according to the Bayesian posterior distribution ${\Pr[(l_1,l_2,l_3)\mid (r_1,r_2,\po(q_1,q_2,q_3))]}$. 

\item\label{Q2} Given information on the average numbers $\langle r_1\rangle$ and $\langle r_2\rangle$ of binding events of receptors $R_1$ and $R_2$, obtain samples  over populations $(l_1,l_2,l_3)$ of the ligand species according to the Bayesian posterior distribution ${\Pr[(l_1,l_2,l_3)\mid (\langle r_1\rangle, \langle r_2\rangle, \po(q_1,q_2,q_3))]}$.
\end{enumerate}
\end{example}

We investigate these questions for arbitrary numbers of receptors and ligands, arbitrary observation matrices $O$, and arbitrary product-Poisson rate parameters $q$, and make the following new contributions:
\begin{bullets}
\item In Section~\ref{sec:probstmt}, we precisely state our question in the general setting. In Section~\ref{sec:ex}, we illustrate our main ideas on Example~\ref{ex:run}.

\item In Section~\ref{sec:scheme}, we describe a reaction network scheme $\eproj$ that takes as input an observation matrix $O$ and outputs a prime chemical reaction network. Our proposed reaction networks have the following merits that make them promising candidates for molecular implementation. Implementing the reactions requires only thermodynamic control and not kinetic control because the reaction rate constants need only be specified upto the equilibrium constant for the reactions (Remark~\ref{rmk:rates}). Our scheme avoids catalysis, and so is robust to ``leak reaction'' situations~\cite{seesawgates} (Remark~\ref{rmk:prime}).

\item In Section~\ref{sec:q1}, we address Question~\ref{Q1}. We show that for each fixed $O$ and $q$, when the chemical reaction system is initialized as prescribed according to the numbers $r_i$ of binding events of receptors, and allowed to evolve according to stochastic mass-action kinetics, then the system evolves towards the desired Bayesian posterior distribution (Theorem~\ref{thm:canonical}).

\item In Section~\ref{sec:q2}, we address Question~\ref{Q2}. We show that for each fixed $O$ and $q$, when the chemical reaction system is initialized as prescribed according to the average numbers $\langle r_i\rangle$ of binding events of receptors, and allowed to evolve according to deterministic mass-action kinetics, then the distribution of unit-volume aliquots of the system evolves towards the desired Bayesian posterior distribution (Theorem~\ref{thm:grandcanonical}).

\item We do a literature review in Section~\ref{sec:rw}, comparing our scheme with other reaction network schemes that process information. Exploiting inherent stochasticity and free energy minimization appear to be the two key new ideas in our scheme. 

\item In Section~\ref{sec:cfw}, we discuss limitations and directions for future work, including a reaction scheme for the expectation-maximization algorithm, which is a commonly used algorithm in machine learning and may be a more sophisticated way for an artificial cell to infer its environment from partial observations.
\end{bullets}

\section{Background}\label{sec:rntrev}
\subsection{Probability and Statistics}
For $n\in\mathbb{Z}_{> 0}$, following~\cite{Amari2016}, KL Divergence $D: \mathbb{R}_{\geq 0}^{n} \times \mathbb{R}_{\geq 0}^{n} \rightarrow \mathbb{R}$ is the function 
\[
D(x\|y)\coloneqq \sum_{i=1}^{n}{x_{i}\log (\frac{x_{i}}{y_{i}}) -x_{i} + y_{i}}
\] 
with the convention $0\log 0 = 0$  and for $p>0$, $p\log 0 = -\infty$. If $x,y$ are probability distributions then $\sum_{i=1}^{n} -x_{i} + y_{i} = 0$ and KL Divergence is the same as relative entropy $\sum_{i=1}^{n}{x_{i}\log \left(\frac{x_{i}}{y_{i}}\right)}$. When the index $i$ takes values over a countably infinite set, we define KL Divergence by the same formal sum as above, and understand it to be well-defined whenever the infinite sum converges in $[0,\infty]$. For $x\in\mathbb{R}^k_{>0}$, by  $\po(x)$ we mean $\Pr[n_1,n_2,\dots,n_k \mid x] = \prod_{i=1}^k \e^{-x_i} \frac{x_i^{n_i}}{n_i!}$. The following lemma is well-known and easy to show.
\begin{lemma}\label{thm:poisdiv}
$D(\po(x)\|\po(y))=D(x\|y)$ for all $x,y \in \mathbb{R}_{>0}^{k}$.
\end{lemma}

The Exponential-Projection or \textbf{E-Projection}~\cite{Amari2016} (or Information-Projection or I-Projection~\cite{csiszar2004information}) of a probability distribution $q$ onto a set of distributions $P$ is $p^*=\arg\min_{p\in P} D(p\|q)$. The Mixture-Projection or \textbf{M-Projection} (or reverse I-projection) of a probability distribution $p$ onto a set of distributions $Q$ is $q^*=\arg\min_{q\in Q} D(p\|q)$.

\subsection{Reaction Networks}
We recall notation, definitions, and results from reaction network 
theory~\cite{feinberg72chemical,horn72necessary,Fein79,Manoj_2011Catalysis,anderson2010product}. For $x,y\in\mathbb{R}^k$, by $x^y$ we mean $\prod_{i=1}^kx_i^{y_i}$, and by $\e^x$ we mean $\prod_{i=1}^k\e^{x_i}$. For $m\in\mathbb{Z}_{\ge 0}^k$, by $m!$ we mean $\prod_{i=1}^k m_i!$.

Fix a finite set $S$ of \textbf{species}. By a \textbf{reaction} we mean a formal chemical equation
\[
	\sum_{i\in S} y_i X_i \rightarrow \sum_{i\in S}y_i'X_i
\]
where the numbers $y_i,y_i'\in\mathbb{Z}_{\geq 0}$ are the \textbf{stoichiometric coefficients}. This reaction is also 
written as $y\to y'$ where $y,y'\in\mathbb{Z}^{\S}_{\geq 0}$. A \textbf{reaction network} is a pair $(\S,\R)$ where $S$ 
is finite, and $\R$ is a finite set of reactions. It is \textbf{reversible} iff for every reaction $y\to y'\in\R$, the 
reaction $y'\to y\in\R$. Fix $n,n'\in\mathbb{Z}^S_{\geq 0}$. We say that $n\mapsto_\R n'$, read $n$ \textbf{maps to} $n'$ 
iff there exists a reaction $y\to y' \in\R$ with $y_i\leq n_i$ for all $i\in S$ and $n' = n + y'-y$. We say that $n
\Rightarrow_\R n'$, or in words that $n'$ is $\R$-\textbf{reachable} from $n$, iff there exist a nonnegative integer $k
\in\mathbb{Z}_{\geq 0}$ and $n(1),n(2),\dots,n(k)\in\mathbb{Z}^\S_{\geq 0}$ such that $n(1)=n$ and $n(k)=n'$ and for $i=1$
to $k-1$, we have $n(i)\mapsto_\R n(i+1)$. A reaction network $(\S,\R)$ is \textbf{weakly reversible} iff for every reaction $y\to y'\in\R$, we have $y'\Rightarrow y$. Trivially, every reversible reaction network is weakly reversible. The \textbf{reachability class} of $n_0\in\mathbb{Z}^\S_{\geq 0}$ is the set $\Gamma(n_0)=\{ n \mid n_0\Rightarrow_\R n\}$. The \textbf{stoichiometric subspace} $H_\R$ is the real span of the vectors $\{y'-y \mid y\to y'\in \R\}$. The \textbf{conservation class} containing $x_0\in\mathbb{R}^\S_{\geq 0}$ is the set $C(x_0)=(x_0 + H_\R)\cap\mathbb{R}^S_{\geq 0}$.

Fix a weakly reversible reaction network $(\S,\R)$. Let $x=(x_i)_{i\in \S}$. The \textbf{associated ideal} $I_{(\S,\R)}\subseteq \mathbb{C}[x]$  is the ideal generated by the binomials $\{ x^y - x^{y'}\mid y\to y'\in\R\}$. A reaction network is \textbf{prime} iff its associated ideal is a prime ideal, i.e., for all $f, g\in \mathbb{C}[x]$, if $fg\in I$ then either $f\in I$ or $g\in I$. 

A \textbf{reaction system} is a triple $(\S,\R,k)$ where $(\S,\R)$ is a reaction network and $k:\R\to\mathbb{R}_{> 0}$ is called the \textbf{rate function}. It is \textbf{detailed balanced} iff it is reversible and there exists a point $q\in\mathbb{R}^\S_{>0}$ such that for every reaction $y\to y'\in\R$:
\[
k_{y\to y'}\,q^y \,(y' - y)   = k_{y'\to y}\, q^{y'}\,(y - y')
\]
A point $q\in\mathbb{R}^\S_{>0}$ that satisfies the above condition is called a \textbf{point of detailed balance}.
 
Fix a reaction system  $(S,\R,k)$. Then \textbf{stochastic mass action} describes a continuous-time Markov chain 
on the state space $\mathbb{Z}^S_{\geq 0}$. A state $n=(n_i)_{i\in S}\in\mathbb{Z}^S_{\ge 0}$ of this Markov chain 
represents a vector of molecular counts, i.e., each $n_i$ is the number of molecules of species $i$ in the population. 
Transitions go from $n \to n+y'-y$ for each $n\in\mathbb{Z}^S_{\geq 0}$ and each $y\to y' \in\R$, with transition rates 
\[
\lambda(n \to n + y'-y) =k_{y\rightarrow y'} \frac{n!}{(n-y)!}
\]

The following theorem states that the stationary distributions of detailed-balanced reaction networks are obtained from products of Poisson distributions. It is well-known, see for example~\cite{whittle1986systems} for a proof.
\begin{theorem}\label{thm:stationary}
If $(\S, \R,k)$ is detailed balanced with $q$ a point of detailed balance then the corresponding stochastic mass action Markov chain admits on each reachability class $\Gamma \subset \mathbb{Z}^{S}_{\ge 0}$ a unique stationary distribution
\[
\pi_\Gamma(n) \propto
\begin{cases}
  \e^{-q}\frac{q^{n}}{n!} & \qquad \text{for } n \in \Gamma \\
  0 & \qquad \text{otherwise}
\end{cases}
\]
\end{theorem}

\textbf{Deterministic mass action} describes a system of ordinary differential equations in {\em concentration} variables $\{x_i(t) \mid i\in \S\}$: 
\begin{equation}\label{eqn:ma}
\dot{x}(t) = \sum_{y\to y' \in \R} k_{y\to y'}\,  x(t)^y \,(y' - y)
\end{equation}

Note that every detailed balance point is a fixed point to Equation~\ref{eqn:ma}. For detailed balanced reaction systems, every fixed point is also detailed balanced. Moreoever, every conservation class $C(x_0)$ has a unique detailed balance point $x^*$ in the positive orthant. Further if the reaction network is prime then $x^*$ is a ``global attractor,'' i.e., all trajectories starting in $C(x_0)\cap\mathbb{R}^S_{>0}$ asymptotically reach $x^*$. (Recently Craciun~\cite{craciun2015toric}  has proved the global attractor theorem for all detailed-balanced reaction systems with a much more involved proof. We do not need Craciun's theorem, the special case which holds for prime detailed-balanced reaction systems and is much easier to prove, suffices for our purposes.) The following Global Attractor Theorem for Prime Detailed Balanced Reaction Systems follows from \cite[Corollary~4.3, Theorem~5.2]{Manoj_2011Catalysis}. See \cite[Theorem~3]{gopalkrishnan2016scheme} for another restatement of this theorem.

\begin{theorem}\label{thm:gac}
Let $(\S,\R,k)$ be a prime, detailed balanced reaction system with point of detailed balance $q$. Fix a point $x_0\in\mathbb{R}^S_{>0}$. Then there exists a point of detailed balance $x^*$ in $C(x_0)\cap\mathbb{R}^S_{>0}$ such that for every trajectory $x(t)$ to Equation~\ref{eqn:ma} with initial conditions $x(0)\in C(x_0)\cap\mathbb{R}^S_{\geq 0}$, the limit $\lim_{t\to\infty} x(t)$ exists and equals $x^*$. Further $D(x(t)\|q)$ is strictly decreasing along non-stationary trajectories and attains its unique minimum value in $C(x_0)\cap\mathbb{R}^S_{\geq 0}$ at $x^*$.
\end{theorem}

\section{Problem Statement}\label{sec:probstmt}
We argue in the next lemma that a product of Poisson distributions is not an unreasonable form to use as a prior on ligand populations. The ideas are familiar from statistical mechanics as well as stochastic processes. We recall them in a chemical context.

\begin{lemma}\label{lem:productpoisson}
Consider a well-mixed vessel of infinite volume with $n$ species $X_{1}, X_{2},\ldots, X_{n}$ at concentrations $x_{1}, x_{2},\dots, x_{n}$ respectively. Assume that the solution is sufficiently dilute, and that molecule volumes are vanishingly small. A unit volume aliquot is taken. Then the probability of finding the population in the aliquot in state $(m_{1}, m_{2},\ldots, m_{n})\in\mathbb{Z}_{\geq 0}$ is given by the product-Poisson distribution $\prod _{ i=1 }^{n}{ \frac{e^{-x_{i}}x_{i}^{m_{i}}}{m_{i}!} } $
\end{lemma}
\begin{proof}
We will first do the analysis for a finite volume $V$ and then let $V\to\infty$. 

Consider a container of finite volume V, which contains species $ X_1, X_2, \ldots,X_n $ at concentrations $x_{1}, x_{2}\ldots, x_{n}$. Consider a unit volume aliquot within this particular container. The probability of finding a particular molecule from the vessel within the unit volume aliquot is $\frac{1}{V}$. The number of molecules of species $X_i$ in the vessel is $Vx_i$ for $i = 1\dots n$. Hence the probability of finding $m_i$ molecules of species $X_i$  in the aliquot is given by the binomial coefficient 
\[
{{V x_{i}}\choose{m_i}} \left({\frac{1}{V}}\right)^{m_i}\left(1-{\frac{1}{V}}\right)^{Vx_i - m_i}.
\]
We assume that the solution is sufficiently dilute, and that molecular sizes are vanishingly small, so that the probability of finding one molecule in the aliquot is independent of the probability of finding a different molecule in the aliquot. This assumption leads to: 
\begin{align*}
\Pr(m_1, m_2, \ldots, m_n\mid x_1, x_2,\dots, x_n) &= \prod_{i=1}^{n}{ {{V x_{i}}\choose{m_i}} \left(\frac{1}{V}\right) ^{m_i} \left( 1-\frac{1}{V}\right) ^{Vx_i - m_i} } 
\end{align*}
The RHS follows because for all $i\in\{1,2,\dots,n\}$:
\begin{align*}
\lim_{V\to\infty} {{V x_{i}}\choose{m_i}} \left(\frac{1}{V}\right) ^{m_i} \left( 1-\frac{1}{V}\right) ^{Vx_i - m_i}
&= \lim_{V\to\infty}\frac{Vx_i(Vx_i-1)\dots(Vx_i - m_i + 1)}{V^{m_i} m_i!}\left[\left(1 - 1/V)\right)^V \right]^{x_i-m_i/V}\
\end{align*}
which equals $\e^{-x_i} x_i^{m_i}/m_i!$
\end{proof}

Fix positive integers $n_R,n_L\in\mathbb{Z}_{\geq 0}$ with $n_R \leq n_L$ denoting the number of receptor species and the number of ligand species respectively. Fix $q = (q_1,q_2,\dots, q_{n_L})\in \mathbb{R}_{> 0}^{n_L}$ denoting Poisson rate parameters for the product-Poisson distribution
$\po(q)$ which we consider as a prior over ligand numbers. Fix an $n_R\times n_L$ \textbf{observation matrix} $O$ with entries $o_{ij}$ in the nonnegative rational numbers $\mathbb{Q}_{\geq 0}$. The entry $o_{ij}$ denotes the affinity of the $i$'th receptor $R_i$ for the $j$'th ligand $L_j$. The intuition is that when ligand $j$ encounters receptor $i$, the propensity that a binding occurs is proportional to $o_{ij}$. So a high-affinity ligand will trigger a receptor more often than a low-affinity ligand with the same concentration, with the number of times they trigger the receptor in proportion to their corresponding entries in the observation matrix. 

Our results in this paper will hold for a subclass of observation matrices which we term tidy. An observation matrix $O=(o_{ij})_{n_R\times n_L}$ is \textbf{tidy} iff for each receptor $R_i$ there exists a \textbf{message vector} $m_i\in\mathbb{Z}^{n_L}_{\geq 0}$ such that $O m_i = e_i$ where $e_i\in\mathbb{R}^{n_R}$ is the unit vector with a $1$ in the row corresponding to the $i$'th receptor. The intuition is that for $j=1$ to $n_L$, species $X_j$ will be the cell's internal representation of the ligand $L_j$. Every time receptor $R_i$ is bound, it will trigger a cascade leading to the synthesis inside the cell of $m_{ij}$ molecules of species $X_j$ for $j=1$ to $n_L$.

Note that there could be multiple message vector sets $\{m_i\}_{i=1\text{ to }n_R}$, so the cell need not choose the ``correct'' one. The task of figuring out the true state of the environment will be left to the reaction network operating inside the cell between the molecules $X_1,X_2,\dots,X_{n_L}$. The messages only perform the task of initializing the reaction network in the right reachability class. The following questions concern us.
\begin{enumerate}
\item\label{Q1} Given information on the exact numbers $r=(r_1,r_2,\dots, r_{n_R})\in\mathbb{Z}_{\geq 0}^{n_R}$ of receptor binding events, obtain samples over populations $l=(l_1,l_2,\dots,l_{n_L})\in\mathbb{Z}_{\geq 0}^{n_L}$ of the ligand species according to the Bayesian posterior distribution 
${\Pr[l\mid (r,\po(q_1,q_2,\dots,q_{n_L}))]}$

\item\label{Q2} Given information on the average numbers $\langle r\rangle =(\langle r_1\rangle,\langle r_2\rangle,\dots, \langle r_{n_R}\rangle)\in\mathbb{R}_{>0}^{n_R}$ of receptor binding events (averaged over the surface of the cell, or time, or both), obtain samples  over populations $l=(l_1,l_2,\dots,l_{n_L})\in\mathbb{Z}_{\geq 0}^{n_L}$ of the ligand species according to the Bayesian posterior distribution ${\Pr[l\mid (\langle r\rangle, \po(q_1,q_2,\dots,q_{n_L}))]}$
\end{enumerate}

\section{An Example}\label{sec:ex}
Before moving to the general solution, we illustrate our main ideas with an example.
\begin{example}[continues=ex:run]
Consider the observation matrix
\[
O=\begin{blockarray}{cccc}
&L_1 &L_2 &L_3\\
\begin{block}{c(ccc)}
R_1 &1 &0 &1\\
R_2 &0 &1 &1\\
\end{block}
\end{blockarray}
\] and the point $q=(q_1,q_2,q_3)\in\mathbb{R}_{>0}^3$ from Example~\ref{ex:run}. We describe a chemical reaction system $(\eproj(O,B),k_q)$ as follows. There is one chemical species $X_i$ corresponding to each ligand $L_i$, so that the species are $X_1,X_2$, and $X_3$. To describe the reactions, we compute a basis for the right kernel of $O$. In this case, the vector $(1,1,-1)^T$ is a basis for the right kernel. (To be precise, we will view the right kernel as a free group in the integer lattice, and take a basis for this free group. This ensures not only that each basis vector has integer coordinates, but also that the corresponding reaction network is prime, which we use crucially in our proofs.) Each basis vector is written as a reversible reaction, with negative numbers representing stoichiometric coefficients on one side of the chemical equation, and positive numbers representing stoichiometric coefficients on the other side. So the vector $(1,1,-1)^T$ describes the reversible pair of reactions $X_1 + X_2 \rightleftharpoons X_3$. 

The rates of the reactions need to be set so that $q$ is a point of detailed balance. For this example, calling the forward rate $k_1\in\mathbb{R}_{>0}$ and the backward rate $k_2\in\mathbb{R}_{>0}$, the balance condition is $k_1 q_1q_2 = k_2 q_3$ so that $k_1/k_2 = \frac{q_3}{q_1q_2}$. One choice satisfying this condition is $k_1 = q_3$ and $k_2=q_1q_2$. Note that our scheme requires only the ratio of the rates to be specified (Remark~\ref{rmk:rates}).

\paragraph{Solution to Question~\ref{Q1}:} Given $r=(r_1,r_2)\in\mathbb{Z}_{\geq 0}^2$ interpreted as $(r_1,r_2)^T=O (l_1,l_2,l_3)^T$, we want to draw samples from the conditional distribution $\Pr[(l_1,l_2,l_3)\mid (r_1,r_2,\text{Poisson}(q_1,q_2,q_3))]$. The statistical solution is to multiply the Bayesian prior $\text{Poisson}(q_1,q_2,q_3)$ by the likelihood $\Pr[(r_1,r_2)\mid (l_1,l_2,l_3,\text{Poisson}(q_1,q_2,q_3))]$, and normalize so probabilities add up to $1$. The likelihood is the characteristic function of the set 
\[L=\{l=(l_1,l_2,l_3)\in\mathbb{Z}^3_{\geq 0} \mid O l^T = r^T\}.
\] 

Note that $O$ is tidy with message vectors $m_1 = (1,0,0)^T$ and $m_2=(0,1,0)^T$. The reaction system $(\eproj(O,B),k_q)$ which is $X_1 + X_2 \xrightleftharpoons[q_1q_2]{q_3} X_3$ here, is initialized at $n(0)=(r_1,r_2,0) = \sum_i r_i m_i$, and allowed to evolve according to stochastic mass-action kinetics with master equation:
\begin{align*}
\dot{p}(n,t) = &p(n_1-1,n_2-1,n_3+1,t)\left(\frac{q_1q_2}{q_3}(n_3+1) - n_1n_2\right)\
\\ +\,\, &p(n_1+1,n_2+1,n_3-1,t)\left((n_1+1)(n_2+1) - \frac{q_1q_2}{q_3}n_3\right)
\end{align*}
where $p(n,t)$ is the probability that the system is in state $n$ at time $t$.
We claim that the steady-state distribution is the required Bayesian posterior. First note that this reaction system has a detailed balanced point $q$, so it admits $\text{Poisson}(q)$ as a steady-state distribution. Since $n(0)\in L$, it is enough to show that $L$ forms an irreducible component of the Markov chain. Together we conclude that the steady-state distribution will be a restriction of $\text{Poisson}(q_1,q_2,q_3)$ to the set $L$. 

To obtain that $L$ forms an irreducible component of the Markov chain, we will crucially use the fact that we chose a basis of the free group to generate our reactions, and not just a basis of the real vector space. This will allow us to prove that the corresponding reaction network is prime, and hence that $L$ forms an irreducible component. Note, for example, that if we had chosen the vector $(2,2,-2)^T$ in the kernel instead of $(1,1,-1)^T$, that would have given us the reaction $2X_1 + 2X_2 \rightleftharpoons 2X_3$ in which case $L$ does not form an irreducible component of the Markov chain since each reaction conserves parity of molecular counts.

\paragraph{Solution to Question~\ref{Q2}:} Given $\langle r\rangle =(\langle r_1\rangle, \langle r_2\rangle)\in\mathbb{R}_{> 0}^2$ of binding events of receptors $R_1$ and $R_2$, with $\langle r\rangle $ interpreted as empirical average of $O (l_1,l_2,l_3)^T$ over a large number of samples of $(l_1,l_2,l_3)$, we want to draw samples from the conditional distribution ${\Pr[(l_1,l_2,l_3)\mid (\langle r_1\rangle, \langle r_2\rangle, \po(q_1,q_2,q_3))]}$. Note that we are conditioning over an event whose probability tends to $0$ unless $Oq^T=\langle r\rangle^T$, so the conditional distribution needs to be defined using the notion of regular conditional distribution~\cite{dembo2010large}. As the number of samples goes to infinity, by the conditional limit theorem~\cite[Theorem~7.3.8, Corollary~7.3.5]{dembo2010large}, this conditional distribution converges to $\text{Poisson}(x^*)$ where $x^*=(x_1^*,x_2^*,x_3^*)\in\mathbb{R}^3_{\geq 0}$ minimizes $D(x\|q)$ among all $x$ satisfying $Ox^* = \langle r\rangle$. Because these results are stated in the reference in much greater generality, to show that these results actually apply to our case will need some technical work which is the content of Section~\ref{sec:q2}.

To compute $x^*$, we allow $(\eproj(O,B),k_q)=X_1 + X_2 \xrightleftharpoons[q_1q_2]{q_3} X_3$ to evolve according to deterministic mass-action kinetics starting from $x(0)=(\langle r_1\rangle,\langle r_2\rangle,0)=\sum_i \langle r_i\rangle m_i$.
\begin{align*}
\left(\begin{array}{c}
\dot{x_1}(t)\
\\\dot{x_2}(t)\
\\\dot{x_3}(t)
\end{array}\right)= \left(x_1(t)x_2(t)-\frac{q_1q_2}{q_3}x_3(t)\right)
\left(\begin{array}{r} 
-1\
\\-1\
\\1\
\end{array}\right)
\end{align*}

Then the equilibrium concentration is the desired $x^*$ by Theorem~\ref{thm:gac}. The required sample can be drawn by sampling a unit aliquot, as in Lemma~\ref{lem:productpoisson}.

Our scheme suggests that the reactions are carried out in infinite volume, which seems impractical. In practise, infinite volume need not be necessary because the chemical dynamics of even molecular numbers as small as $50$ molecules are often described fairly accurately by the infinite-volume limit. Further, our scheme suggests an infinite number of samples for this to work correctly, which also looks impractical. However, the rate of convergence is exponentially fast, so the scheme can be expected to work quite accurately even with a moderate number of samples. Analysis beyond the scope of the current paper is needed to explore the tradeoffs in volume and number of samples (also see Section~\ref{sec:cfw}).
\end{example}

\section{Main}
\subsection{A Reaction Scheme}\label{sec:scheme}
In this subsection, we present a reaction scheme $\eproj$ (short for projection) that takes as input a matrix $O$ with rational entries, and a basis $B$ for the free group $\ker O\cap\mathbb{Z}^{n_L}_{\geq 0}$ and outputs a reversible reaction network $\eproj(O,B)$ that is prime. The same scheme, appropriately initialized, serves to perform M-projection (as we showed in \cite{gopalkrishnan2016scheme}) and E-projection, as we show here.

\begin{definition} Fix a matrix $O=(o_{ij})_{m\times nL}$ with rational entries $o_{ij}\in\mathbb{Q}$, and a basis $B$ for the free group $\mathbb{Z}^{n}\cap\ker O$. The reaction network $\eproj(O, B)$ is described by species $X_1,X_2,\ldots,X_{n}$ and for each $b\in B$ the reversible reaction:
$
\sum_{j:b_j>0}b_jX_j \rightleftharpoons\sum_{j:b_j<0}-b_jX_j
$
\end{definition}

\begin{remark}\label{rmk:rates}
Exquisitely setting the specific rates of individual reactions to desired values requires a detailed understanding of molecular dynamics, and is forbiddingly difficult with current molecular technology. When we set rates, we will only require that a given point remains a point of detailed balance. This is equivalent to specifying the equilibrium constants of all the reactions. This is an equilibrium thermodynamics condition, hence much less forbidding.
\end{remark}

\begin{lemma} \label{thm:prime1}
Fix a matrix $O=(o_{ij})_{m\times n}$ with rational entries $o_{ij}\in\mathbb{Q}$, and a basis $B$ for the free group $\mathbb{Z}^n\cap\ker O$. Then the reaction network $\eproj(O,B)$ is prime.
\end{lemma}
\begin{proof}
\cite[Corollary~1.15]{miller2011theory} establishes this when $O$ is a matrix of integers. Scaling the rational entries to make them all integers makes no difference to the kernel.
\end{proof}

\begin{remark}\label{rmk:prime}
From \cite[Theorem 5.2]{Manoj_2011Catalysis}, prime reaction networks are free of catalysis. Catalysts require care to implement. Ideally a catalyst should act as a switch, so that its absence completely shuts off the catalyzed reaction. In practice, there is always a ``leak reaction''~\cite{seesawgates} even in the absence of the catalyst species. Care needs to be taken that the timescales of the leak are much slower than the timescales of the catalyzed reaction to get an acceptable approximation to the final answer. It is therefore notable that our scheme is able to perform a nontrivial computation even though it admits an implementation wholly free of catalysis.
\end{remark}

\begin{example}
Consider the reaction $2X\rightleftharpoons 0$. On the state space $\mathbb{Z}_{\geq 0}$, this reaction will preserve the parity of the initial number $n_0$ of $X$. This is a case where the intersection of a conservation class $C(n_0)$ with the state space does not equal the reachability class $\Gamma(n_0)$. It turns out that these ``non-benign'' situations only happen when the reaction network is not prime. We will use this property when answering Questions 1 and 2, so we establish it now.
\end{example}

\begin{definition}
A weakly-reversible reaction network $(\S,\R)$ is \textbf{benign} iff for all $n_0\in\mathbb{Z}^\S_{\ge 0}$, the conservation class $C(n_0)\cap\mathbb{Z}^\S_\geq 0 = \Gamma(n_0)$, the reachability class of $n_0$.
\end{definition}

\begin{lemma}\label{lem:benign}
Every prime reaction network is benign.
\end{lemma}
\begin{proof}
Let $(\S,\R)$ be a prime reaction network. This means that the associated ideal $(x^y - x^{y'})_{y\to y'\in\R}$ is prime. We define the \textbf{associated lattice} as 
\[
\mathcal{L} = \left\{ \sum_{y\to y'\in\R}  a_{y\to y'} (y'-y) \mid a_{y\to y'}\in\mathbb{Z}\text{ for all }y\to y'\in\R\right\}.
\]
Note from \cite{miller2011theory} that $\mathcal{L}$ is \textbf{saturated}, i.e., if $k\in\mathbb{Z}$ and $v\in \mathbb{Z}^\S$ are such that $k v\in \mathcal{L}$ then $v\in \mathcal{L}$

Suppose $n_0,n_0' \in\mathbb{Z}^\S_{\geq 0}$ such that $n_0'\in C(n_0)$ but $n_0'$ is not reachable from $n_0$. 
The condition $n_0'\in C(n_0)$ means that there is a rational combination 
\[
n_0' - n_0 = \sum_{y\to y'\in\R} b_{y\to y'} (y' -y)
\]
This shows that for some sufficiently large integer $M$, the quantity $M(n_0'-n_0)\in \mathcal{L}$. Since $\mathcal{L}$ is saturated, $n_0'-n_0\in \mathcal{L}$. Hence there is an integer combination 
\[
n_0' - n_0 = \sum_{y\to y'\in\R} c_{y\to y'} (y' -y).
\]
Since $(\S,\R)$ is weakly-reversible, there is a path $y'\Rightarrow_\R y$ for every $y\to y'\in\R$, and therefore there is a combination over nonnegative integers. This implies that $n_0\Rightarrow_\R n_0'$. Hence the network is benign.
\end{proof}

\subsection{Solution to Question 1}\label{sec:q1}
In this section we solve Question~\ref{Q1} using the reaction network $\eproj(O,B)$.

Fix an $n_L\times n_R$ tidy observation matrix $O=(o_{ij})_{n_R\times n_L}$ with non-negative rational entries $o_{ij}\in\mathbb{Q}_{\ge 0}$, and message vectors $\{m_i\in\mathbb{Z}_{\ge 0}^{n_L}\}_{i=1,2,\dots,n_R}$, Poisson rate parameter vector $q\in\mathbb{R}_{\ge 0}^{n_L}$, and number $r\in \mathbb{Z}_{\ge 0}^{n_R}$ of receptor binding events observed. Fix a basis $B$ for the free group $\ker O\cap\mathbb{Z}^{n_L}_{\geq 0}$. Let $k_q$ be a function of rate constants for the reaction network $\eproj(O,B)$ such that $q$ is a point of detailed balance of the reaction system $(\eproj(O,B),k_q)$. For example, the choice $k_q(y\to y') = q^{y'}$ satisfies this requirement.
\begin{theorem}\label{thm:canonical}
Consider Stochastic Mass Action for the reaction system $(\eproj(O,B),k_q)$ from the initial state $n(0) =\sum_{i=1}^{n_R}r_im_i$. Then the Bayesian Posterior ${\Pr[l\mid (r,\po(q))]}$ is the stationary distribution of this Markov chain.
\end{theorem}
\begin{proof}
  Let $L = \left\{l \in \mathbb{Z}_{\geq 0}^{n_{L}} \mid Ol = r\right\}$. From Bayes Theorem $\Pr[l\mid (r,\po(q))]\propto \text{Prior}\times\text{Likelihood}$. The prior is $\po(q)$ and the likelihood is $\Pr[r\mid l]=\Pr[Ol=r]$ which is the characteristic function on $L$. Therefore
$
\Pr[l\mid (r,\po(q))]\propto
\begin{cases}
\e^{-q}\frac{q^l}{l!} &\qquad \text{for } l \in L\\
0 & \qquad \text{otherwise}
\end{cases}
$
\\Since the reaction network $\eproj(O,B)$ is prime, by Lemma~\ref{thm:prime1} and Lemma~\ref{lem:benign}, $\eproj(O,B)$ is benign. By construction $n(0)\in L$, and so $L$ is the reachability class $\Gamma(n(0))$. Applying Theorem~\ref{thm:stationary} to $L=\Gamma(n_0)$
\[
\pi_L(l) \propto
\begin{cases}
\e^{-q}\frac{q^{l}}{l!} & \qquad \text{for } l \in L \\
0 & \qquad \text{otherwise}
\end{cases}
\]
which is exactly the Bayesian Posterior  $\Pr[l\mid (r,\po(q))]$.
\end{proof}

In the following theorem, we show that our reaction scheme has computed an E-Projection.

\begin{theorem}\label{thm:eproj1}
Let $\mathcal{P}\coloneqq\{\text{Probability measure }P \text{ on } \mathbb{Z}_{\ge 0}^{n_L} \mid P(l) = 0 \text{ for all } l\notin L\}$. Then  $\Pr[l\mid (r,\po(q))]$ is the E-Projection of $\po(q)$ on $\mathcal{P}$.
\end{theorem}
\begin{proof}
The E-projection of $\po(q)$ onto $\mathcal{P}$ is given by $P^* = \operatorname* {arg\,min}_{\substack{P\in\mathcal{P}}}{D(P\|\po(q))}$. We use Lagrange multiplier to minimize $D(P\|\po(q))$ with constraints $\sum_{l\in L} {P(l)} = 1$ and $P(l) = 0$ for $l \notin L$.
\[
F(P,\lambda,\mu) = D(P\|\po(q)(l)) + \lambda\left(\sum_{l\in L} {P(l)} - 1 \right)+ \sum_{l\notin L}{\mu_{l}P(l)}
\]
At $P^*$, $\frac{\partial F}{\partial P(l)}=0$ for all $l\in\mathbb{Z}_{\ge 0}^{n_L}$. That is, $\log\left(\frac{P^*(l)}{\po(q)}\right)+1+\lambda=0$ if $l\in L$ and $P^*(l)=0$ if $l\notin L$. That is,
\[
P^*(l) \propto
\begin{cases}
\po(q)(l) & \qquad \text{for } x \in L \\
0 & \qquad \text{otherwise}
\end{cases}
\]
which is the Bayesian Posterior $\Pr[l\mid (r,\po(q))]$
\end{proof}


\subsection{Solution to Question 2}\label{sec:q2}
In this subsection we solve Question~\ref{Q2} using the reaction network $\eproj(O,B)$. We first characterize the Bayesian Posterior $\Pr[l\mid (\langle r\rangle ,\po(q))]$ as an E-projection using a conditional limit theorem. 

\begin{definition}
Fix $\langle r\rangle\in\mathbb{R}^{n_R}_{>0}$. Then $\mathcal{P}_{\langle r\rangle}$ is the set of those probability measures on $\mathbb{Z}_{\ge 0}^{n_L}$ such that if $Y$ is a random variable distributed according to $P\in\mathcal{P}_{\langle r\rangle}$ then the expected value $\langle OY \rangle_P=\langle r \rangle$.
\end{definition}

\begin{theorem}\label{thm:corgcp}
Fix $\langle r\rangle\in\mathbb{R}^{n_R}_{>0}$. Then $\Pr[l\mid (\langle r\rangle,\po(q))]$ is a Poisson distribution, as well as the E-Projection $\arg\min_{P\in\mathcal{P}_{\langle r\rangle}}D(P\|\po(q))$ of $\po(q)$ on $\mathcal{P}_{\langle r\rangle}$.
\end{theorem}
\begin{proof}
We apply the Gibbs Conditioning Principle (\cite[Theorem~7.3.8]{dupuis2011weak}) $n_R$ times with a sequence of energy functions $U_1,\dots,U_{n_R}$ which iteratively set the expected values of the $n_R$ rows of $O$ to the corresponding values from $\langle r \rangle$. The intuition is that this is a formal way of doing Lagrange optimization.

To show that this result can be applied, we choose the space $\Sigma$ as $\mathbb{R}^{n_L}$, the initial distribution $\mu=\mu_0$ as $\po(q)$ on $\mathbb{Z}^{n_L}_{\geq 0}$ and $0$ everywhere else, and for $i=1$ to $n_R$, we define the function $U_i:\Sigma\to [0,\infty)$ by $
U_i(n)=\frac{(On)_i}{\langle r_i\rangle}$. The sequence of Gibbs distributions are then defined by 
$
\frac{d\mu_{i+1}}{d\mu_i} = \frac{\e^{-\beta_i U_i(n)}}{Z_{\beta_i}}
$
where $Z_{\beta_i}$ is the normalizing constant. It is easily checked that each of these is a Poisson distribution since the $U_i$'s are linear functions. Since $\langle r\rangle\in\mathbb{R}^{n_R}_{>0}$, there is nonzero probability under $\mu_{i-1}$ that $(Ox)_i<\langle r_i\rangle$ for all $i$. Hence for $i=1$ to $n_R$ it follows that $\mu_{i-1}(\{x \mid U_i(x)<1\})>0$. The other condition $\mu_{i-1}(\{x \mid U(x)>1\})>0$ is true since under a Poisson distribution, $(Ox)_i$ can take arbitrarily large integer values with nonzero probability. Since the $\mu_i$ are all Poisson, $\beta_\infty = -\infty$ since Poisson distributions converge for arbitrarily small nonegative values of rate parameters. Hence the assumptions of \cite[Lemma~7.3.6]{dupuis2011weak} are satisfied and we get to apply \cite[Theorem~7.3.8]{dupuis2011weak} sequentially $n_R$ times and conclude that the empirical distribution on the space $\mathbb{Z}^{n_L}_{\geq 0}$ converges weakly to a Poisson distribution $\mu_{n_R} = \po(p^*)\in\mathcal{P}_{\langle r\rangle}$, which is also the E-projection $\arg\min_{\substack{P\in\mathcal{P}_{\langle r\rangle}}}D(P\|\po(q))$. 
\end{proof}

Now fix an $n_L\times n_R$ tidy observation matrix $O=(o_{ij})_{n_R\times n_L}$ with non-negative rational entries $o_{ij}\in\mathbb{Q}_{\ge 0}$, and message vectors $\{m_i\in\mathbb{Z}_{\ge 0}^{n_L}\}_{i=1,2,\dots,n_R}$, Poisson rate parameter vector $q\in\mathbb{R}_{\ge 0}^{n_L}$, and average number  $\langle r\rangle\in \mathbb{R}_{> 0}^{n_R}$ of receptor binding events observed. Fix a basis $B$ for the free group $\ker O\cap\mathbb{Z}^{n_L}_{\geq 0}$. Let $k_q$ be a function of rate constants for the reaction network $\eproj(O,B)$ such that $q$ is a point of detailed balance of the reaction system $(\eproj(O,B),k_q)$. For example, the choice $k_q(y\to y') = q^{y'}$ satisfies this requirement.

\begin{theorem}\label{thm:grandcanonical}
Consider the solution $x(t)$ to the Deterministic Mass Action ODEs for the reaction system $(\eproj(O,B),k_q)$ from the initial concentration $x(0) =\sum_{i=1}^{n_R} \langle r_i\rangle m_i$. Let $x^* = \lim_{t\to\infty} x(t)$. Then $x^*$ is well-defined, and the Bayesian Posterior ${\Pr[l\mid (r,\po(q))]}$ equals $\po(x^*)$. That is, one obtains samples from the Bayesian Posterior by measuring the state of a unit volume aliquot of the system at equilibrium.
\end{theorem}
\begin{proof}
Note that $\po(x(0)) \in\mathcal{P}_{\langle r\rangle}$. Further the reaction vectors span the kernel of $O$ so we have $x\in C(x(0))\cap\mathbb{R}^{n_L}_{>0}$ iff $\po(x)\in\mathcal{P}_{\langle r\rangle}$. By Theorem~\ref{thm:corgcp}, the distribution ${\Pr[l\mid (r,\po(q))]}$ equals $\po(y)$ for some $y\in\mathbb{R}^{n_L}_{>0}$. Further, it is an E-projection so that, among all Poisson distributions in $\mathcal{P}_{\langle r\rangle}$, the relative entropy $D(\po(y) \| q)$ is minimum. By Lemma~\ref{thm:poisdiv}, the E-projection of $\{\po(x)\mid x\in C(x(0))\cap\mathbb{R}^{n_L}_{>0}\}$ to $\po(q)$ is the Poisson distribution of the E-projection of $C(x(0))\cap\mathbb{R}^{n_L}_{>0}$ to $q$.

By Lemma~\ref{thm:prime1}, the reaction network $\eproj(O,B)$ is prime. Further the reaction system $(\eproj(O,B),k_q)$ is detailed balanced with $q$ a point of detailed balance, by assumption. Hence by Theorem~\ref{thm:gac}, the limit $x^*$ is well-defined and is the E-projection of $C(x(0))\cap\mathbb{R}^{n_L}_{>0}$ to $q$. Together we have ${\Pr[l\mid (r,\po(q))]}=\po(x^*)$. We can  sample from a unit aliquot at equilibrium due to Lemma~\ref{lem:productpoisson}.
\end{proof}

\section{Related Work}\label{sec:rw}
Various schemes have been proposed to perform information processing with reaction networks, for example, \cite{seesawgates,qian2011scaling} which shows how Boolean circuits and perceptrons can be built, \cite{Klavins_2011Biomolecular} which shows how to implement linear input/ output systems, \cite{daniel2013synthetic} exploiting analogies with electronic circuits, \cite{buisman2009computing} for computing algebraic functions, etc. Some of these schemes have even been successfully implemented in vitro. 

Each of these schemes has been inspired by analogy with some existing model of computation. However, reaction networks as a computing platform has some unique opportunities and challenges. It is an inherently distributed and stochastic platform. Noise manifests as leaks in catalyzed reactions. We can tune equilibrium thermodynamic parameters, but kinetic-level control is very difficult. In addition, one needs to keep in mind the tasks that reaction networks are called upon to solve in biology, or might be called upon to solve in technological applications. Keeping these factors in mind, there is value in considering a scheme which attempts to uncover the class of problems that is suggested by the mathematical structure of reaction network dynamics.

In trying to uncover such a class of problems, we have looked to the ideas of Maximum Entropy or MaxEnt~\cite{jaynes1957information} which form a natural bridge between Machine Learning and Reaction Networks. The systematic foundations of statistics based on the minimization of KL-divergence (equivalently, free energy) go back to the pioneering work of Kullback~\cite{kullback1997information}. The conceptual, technical, and computational advantages of this approach have been brought out by subsequent workers~\cite{csiszar2004information,Amari2016,cencov2000statistical}. This work has also been put forward as a mathematical justification of Jaynes' MaxEnt principle. Our hope is that those parts of statistics and machine learning that can be expressed in terms of minimization of free energy should naturally suggest reaction network algorithms for their computation. 

The link between statistics/ machine learning and reaction networks has been explored before by Napp and Adams~\cite{napp2013message}. They propose a deterministic mass-action based reaction network scheme to compute single-variable marginals from a joint distribution given as a factor graph, drawing on ``message-passing'' schemes. Our work is in the same spirit of finding more connections between machine learning and reaction networks, but the nature of the problem we are trying to solve is different. We are trying to estimate a full distribution from partial observations. In doing so, we exploit the inherent stochasticity of reaction networks to represent correlations and do Bayesian inference. 

One previous work which has engaged with stochasticity in reaction networks is by Cardelli et al.~\cite{CardelliMarta2016}. They give a reaction scheme that takes an arbitrary finite probability distribution and encodes it in the stationary distribution of a reaction system. In comparison, we are taking samples from a marginal distribution and encoding the full distribution in terms of the stationary distribution. Thus our scheme allows us to do conditioning and inference.

In Gopalkrishnan~\cite{gopalkrishnan2016scheme}, one of the present authors has proposed a molecular scheme to do Maximum Likelihood Estimation in Log-Linear models. The reaction networks employed in that work are essentially identical to the reaction networks employed in this work, modulo some minor technical differences. In that paper, the reaction networks were used to obtain M-projections (or reverse I-projections), and thereby to solve for Maximum Likelihood Estimators. In this paper, we obtain E-projections, and sample from conditional distributions. The results in that paper were purely at the level of deterministic mass-action kinetics. The results in this paper obtain at the level of stochastic behavior. 



\section{Discussions}\label{sec:cfw}
We have shown that reaction networks are particularly well-adapted to perform E-projections. In a previous paper~\cite{gopalkrishnan2016scheme}, one of the authors has shown how to perform M-projections with reaction networks. Intuitively, an E-projection corresponds to a ``rationalist'' who interprets observations in light of previous beliefs, and an M-projection corresponds to an ``empiricist'' who forms new beliefs in light of observations.

Not surprisingly, these two complementary operations keep appearing as blocks in various statistical algorithms. Our two schemes should be viewed together as building blocks for implementing more sophisticated statistical algorithms. For example, the \textbf{EM algorithm} works by alternating E and M projections~\cite{Amari2016}. If our two reaction networks are coupled so that the point $q$ is obtained by the scheme in \cite{gopalkrishnan2016scheme}, and the initialization of the scheme in this paper is used to perturb the conservation class for the M-projection correctly, then an ``interior point'' version of the EM algorithm may be possible, though perhaps not with detailed balance but in a ``driven'' manner reminiscent of futile cycles. 

We have illustrated how E-projections might apply to the situation of an artificial cell trying to infer its environment from partial observations. We are acutely aware that our illustration is far from complete. A more sophisticated algorithm would work in an ``online'' fashion, adjusting its estimates on the fly to each new receptor binding event. This certainly appears within the scope of the kind of schemes we have outlined, but more careful design and analysis is necessary before formal theorems in this direction can be shown. Also we think it likely that the schemes that will prove most effective will work neither purely in the regime of the first scheme, nor purely in the regime of the second scheme, but somewhere in between. How long a time window they average over, and how large a volume is optimal, and how these choices tradeoff between sensitivity and reliability, these are questions for further analysis.

One glaring gap in our narrative is that we require the internal species $X_i$ to be as numerous as the outside ligands $L_i$. A much more efficient encoding of ligand population vectors should be possible, drawing on ideas from graphical models, so that the number of representing species need only be a logarithm of the number of ligands being represented. Moreover it may be possible to perform E and M projections directly on these graphical model representations.

Our constructions and results of Section~\ref{sec:scheme} were carried out for arbitrary matrices with rational number entries. We only used the assumption of ``tidy'' matrices to set initial conditions in Theorems~\ref{thm:canonical}, \ref{thm:grandcanonical}. If some other method of setting initial conditions correctly is available, for example by performing matrix inversions with a reaction network, then the technical condition of tidy matrices can be dropped. In defence of the assumption that our observation matrices are tidy, it is not inconceivable that through evolution a biological cell would have evolved its receptors so that the affinity matrix allows for simple meaningful messages to be transmitted inside the cell. 

Note that the mathematics does not require the restriction of the affinities $o_{ij}$ to \textbf{nonnegative} rational numbers. We could have admitted negative numbers, and all our results would go through.

\bibliographystyle{plain}
\bibliography{sample,eventsystems}

\end{document}